\documentclass{paper}

\usepackage{subfigure}

\usepackage{graphicx}
\usepackage{amssymb}
\usepackage{amsmath}
\usepackage{amsthm}
\usepackage{cite}
\usepackage{fullpage}

\newcommand\R{{\mathbb R}}
\newcommand\N{{\mathbb N}}
\newcommand\eps{{\varepsilon}}
\newcommand\eqdef{:=}
\newcommand\matousek{{Matou{\v{s}}ek}}
\DeclareMathOperator{\lev}{lev}

\newtheorem {theorem} {Theorem}[section]
\newtheorem {prop}[theorem] {Proposition}
\newtheorem {lemma}[theorem] {Lemma}

\title{A Lower Bound for Shallow Partitions}

\author{Wolfgang Mulzer
         \and
        Daniel Werner\thanks{DW was funded by 
	 Deutsche Forschungsgemeinschaft within the Research Training
         Group (Graduiertenkolleg) ``Methods for Discrete Structures''.}
	}
\institution{\texttt{\{mulzer,werner\}@inf.fu-berlin.de}\\\\
        Institut f\"ur Informatik,\\
        Freie Universit\"at Berlin,\\ 14195 Berlin, Germany}

\begin{document}
\maketitle

\begin{abstract}
Let $P$ be a planar $n$-point set. A \emph{$k$-partition} of
$P$ is a subdivision of $P$ into $\lceil n/k\rceil$ parts of
roughly equal size and a sequence of triangles such that
each part is contained in a triangle. A line is \emph{$k$-shallow} 
if it has at most $k$ points of $P$ below it.
The crossing number of a $k$-partition is the maximum number
of triangles in the partition that any $k$-shallow line intersects.
We give a lower bound of $\Omega(\log (n/k)/\log\log(n/k))$ 
for this crossing number, answering a 20-year old question
of \matousek.
\end{abstract}

\section{Introduction}

\emph{Range searching} is a fundamental problem in computational 
geometry that has long driven innovation in the field~\cite{AgarwalEr99}:
given a set of $n$ points in $d$ dimensions, find a data
structure such that all points inside a given query range
can be found efficiently. Depending on the precise
nature of the query range and on the dimension,
many different versions of the problem can be studied. Consequently, 
a wide variety of techniques have been developed to address them. Among these
tools we can find such classics as range trees and
$k$d-trees~\cite[Chapter~5]{deBergChvKrOv08}, 
$\eps$-nets and cuttings~\cite{Chazelle00}, spanning trees with small 
crossing number~\cite{Welzl92}, 
geometric partitions~\cite{Matousek92b}, and many more. For several problems,
almost matching lower bounds are known (in certain models of 
computation)~\cite{Chazelle00}.

\emph{Geometric partitions} provide the most effective means for solving the
\emph{simplex range searching} problem, where the query range is
given by a $d$-dimensional simplex~\cite{Chan10,Matousek92b}.
They provide a way to subdivide a point set into 
parts of roughly equal size, such that (i) each part is contained in a simplex;
and (ii) any given hyperplane intersects only few of these simplices. 
This makes it possible to construct a tree-like data structure 
in which each node corresponds to a simplex in an appropriate geometric partition.
With a careful implementation, one can achieve query time
$O(n^{1-1/d} + z)$ with linear space~\cite{Chan10} (here $z$ is
the output size, i.e., the number of reported points).

If the query simplex degenerates to a half-space, we can
do better~\cite{Matousek92}. For this, we need a more specialized version of
geometric partitions, called \emph{shallow partitions}. Again, these
partitions provide a way for subdividing a $d$-dimensional point set into
parts of roughly equal size, such that each part is contained in a simplex and
such that a hyperplane intersects only few of these simplices. 
This time, however, we restrict ourselves to \emph{shallow} hyperplanes. 
Such hyperplanes have only few points to one side. 
Thus, we only have the guarantee
that any shallow hyperplane will intersect few simplices of the partition
(see below for details).
This makes it possible to decrease the number of simplices that are intersected 
and to achieve better bounds for halfspace range searching. Namely, one
can obtain for $d \geq 4$ a linear-space data structure that answers a query in time
$O(n^{1-1/\lfloor d/2 \rfloor} + z)$, where $z$ is the output 
size~\cite{Chan10} (for $d=2,3$, one can achieve query time $O(\log n + z)$
and linear space~\cite{AfshaniCh09}).

Shallow partitions (as well as their cousins---shallow cuttings) have proved
invaluable tools in computational geometry and have found numerous further
applications. Nonetheless, there still remain some open questions. 
As mentioned above, we would like every shallow hyperplane 
to intersect as few simplices of the shallow partition as possible.
But what exactly is possible? For dimension $d \geq 4$, the original bound
by \matousek~\cite{Matousek92b} is known to be asymptotically tight. For
lower dimensions, however, \matousek\ asked whether his result could be 
improved.
It took almost 20 years until Afshani and Chan~\cite{AfshaniCh09}
provided the first lower bound in three dimensions, almost
matching the upper bound. For the plane, however, so far no nontrivial
lower bounds appear in the literature. 

Here, we will give a construction that
provides such a lower bound for shallow partitions in two dimensions.
Our result almost matches the upper bound and also gives an alternative
proof for the bound of Afshani and Chan~\cite{AfshaniCh09}.
A similar construction has been discovered independently by
Afshani~\cite{Afshani10}.

\section{Shallow partitions}

We begin by providing the details of \matousek's shallow partition theorem
in two dimensions.
Let $P \subseteq \R^2$ be a planar $n$-point set in general position.
Let $k \in \{1, \ldots, n\}$ be a parameter. 
A \emph{$k$-partition} $\mathcal{P}$ for $P$ consists of two parts:
(i) a sequence $P_1$, $P_2$, $\ldots$, $P_{\lceil n/k\rceil}$ of
pairwise disjoint subsets of $P$ such that $\bigcup_i P_i = P$
and $|P_i| = k$ for $i = 1, \ldots, \lfloor n/k \rfloor$; and
(ii) a sequence $\Delta_1$, $\Delta_2$, $\ldots$, $\Delta_{\lceil n/k\rceil}$
of triangles such that $P_i \subseteq \Delta_i$ for all $i$.

Now let  $\ell$  be a line that does not contain any point in $P$, and let
$\ell^-$ denote the open halfplane below $\ell$. 
We say that $\ell$ is \emph{$k$-shallow} if 
$|\ell^- \cap P| \leq k$. Given a $k$-partition $\mathcal{P}$
of $P$, the \emph{crossing number} of $\mathcal{P}$ is the 
maximum number of triangles in $\mathcal{P}$ that are intersected
by any $k$-shallow line. For any given $k$, the goal is to find 
a $k$-partition of $P$ whose crossing number is as small as possible.
\matousek~\cite[Theorem~3.1]{Matousek92} 
proved the following theorem.
\begin{theorem} 
\label{thm:shallowpart}
Let $P$ be a planar $n$-point set in general position and let 
$k \in \{1, \ldots,
n\}$. Then there exists a $k$-partition of $P$ with crossing number
$O(\log(n/k))$. \hfill$\Box$
\end{theorem}

\matousek's original proof uses cuttings and a variant of the 
iterative reweighting technique (also known as the multiplicative weights 
update method~\cite{AroraHaKa10}), and it readily generalizes to higher 
dimensions. 
More recently, Har-Peled and Sharir~\cite[Lemma~3.3]{HarPeledSh11} give
an approach for proving Theorem~\ref{thm:shallowpart} with elementary 
means, but it is not clear whether their technique can be applied to 
higher dimensions. As mentioned in the introduction, 
\matousek~\cite{Matousek92} 
asked whether
the crossing number in Theorem~\ref{thm:shallowpart} can be
improved to $O(1)$. He conjectured that the answer is no. 
Afshani and Chan~\cite{AfshaniCh09} proved that for any $k$ there are
arbitrarily large point sets
in $\R^3$ such that the crossing number
of any $k$-partition for them is 
$\Omega\bigl(\frac{\log (n/k)}{\log\log (n/k)}\bigr)$. 
However, their construction
does not apply for two dimensions. 
Hence, we will describe here a different---and arguably simpler---construction
that yields the same lower bound for the plane.
Independently, Afshani~\cite{Afshani10} used very similar ideas to
obtain the same lower bound.

\section{The Lower Bound}

Let $a(n,k)$ be the minimum crossing number that
a $k$-partition can achieve for any planar $n$-point set
in general position. 
For the lower bound, we shall consider the dual setting. 
We use the standard duality 
transform along the unit paraboloid that maps the point 
$p:(p_x,p_y)$ to the line $p^*: y = 2p_x x -p_y$ and vice
versa~\cite{Mulmuley94}.

A point set $P$ dualizes to a set $P^*$ of planar lines. We now
define the \emph{$k$-level} of $P^*$, $\lev_k(P^*)$~\cite{SharirAg95}. 
It is the closure of the set 
of all points that lie on a line of $P^*$ and 
that have exactly $k$ lines of $P^*$ beneath them. 
We observe that $\lev_k(P^*)$ is an $x$-monotone polygonal curve 
whose edges and vertices come from the arrangement
of $P^*$. 
Let $\mathcal{C}$ be the upper convex hull of $\lev_k(P^*)$.
For each vertex $v$ of $\mathcal{C}$, we let $P^*_{v} \subseteq P^*$ 
denote the set of lines
beneath it. We call $P^*_v$ the \emph{conflict set} of $v$. 
We have $|P^*_{v}| = k$,\footnote{Note 
that $\lev_k(P^*)$ may also contain vertices with
only $k-1$ lines of $P^*$ beneath them, but these vertices
cannot appear on $\mathcal{C}$, since they correspond to
a concave bend in $\lev_{k}(P^*)$.} 
hence $v$ is dual to
a $k$-shallow line $v^*$ in the primal plane.

Now we can interpret shallow partitions in the dual plane:
\begin{prop}\label{prop:dualpartition}
Let $\mathcal{C}$ be an $x$-monotone downward convex chain, and let
$L$ be a set of $n$ lines such that for each
vertex $v$ of $\mathcal{C}$ the conflict set $L_v$ has cardinality $k$.
Then there exists a coloring of $L$ 
such that (i) each color class 
has size at most $k$; and (ii) each conflict set $L_v$ contains
at most $a(n,k)+1$ different colors.
\end{prop}
\begin{proof}
Consider the primal plane, where $L = P^*$ corresponds to a point set $P$.
By assumption, there exists a $k$-partition $\mathcal{P}$ of $P$ 
with crossing number $a(n,k)$. Each vertex $v$ of $\mathcal{C}$ 
corresponds 
to a $k$-shallow line $v^*$, and at most one triangle of $\mathcal{P}$ 
can be wholly contained in $v^{*-}$. Thus, the claim follows from the 
properties of the duality transform.
\end{proof}

We are now ready to describe the construction. 
Let $m = 2^\beta$ be a power of 2 and 
let $\mathcal{C}$ be an $x$-monotone convex chain with $m$ vertices.
We denote these vertices by
$v_1$, $\ldots$, $v_m$, from left to right. 
Now, for $j = 0, \ldots, \beta$, let
$L_j$ be a set of $m/2^j$ lines such that
the first line in $L_j$ lies exactly below the vertices $v_1$ to $v_{2^j}$,
the second line lies below $v_{2^j+1}$ to $v_{2 \cdot 2^{j}}$,
the third line lies below $v_{2 \cdot 2^j+1}$ to $v_{3 \cdot 2^{j}}$,
etc. We set $L' \eqdef \bigcup_{j=0}^{\beta} L_j$. See Fig. \ref{fig:Ls}.
\begin{center}
\begin{figure}[htb] 
  \centering\subfigure{ 
    \includegraphics[scale=.7]{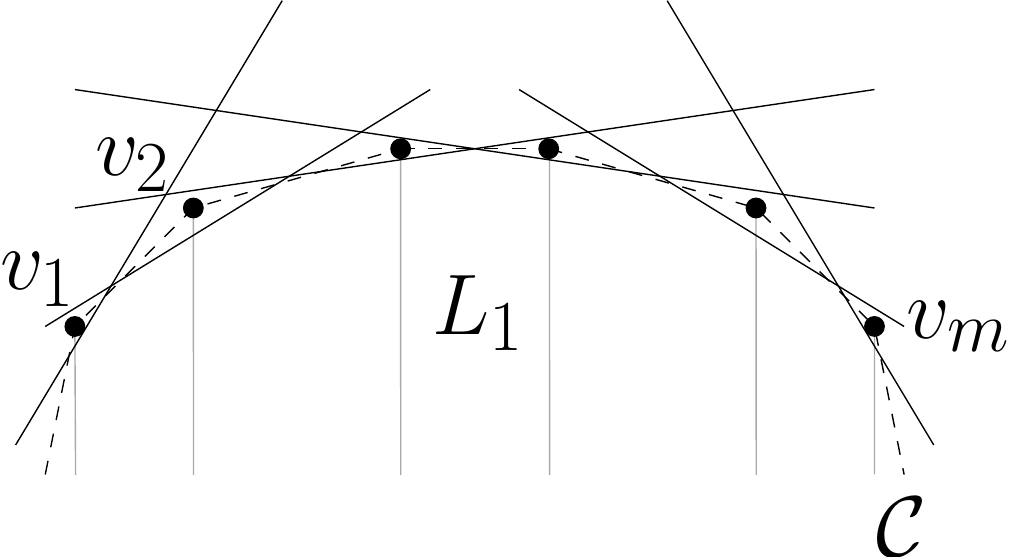}

   } 
  \centering\subfigure{ 
    \includegraphics[scale=.7]{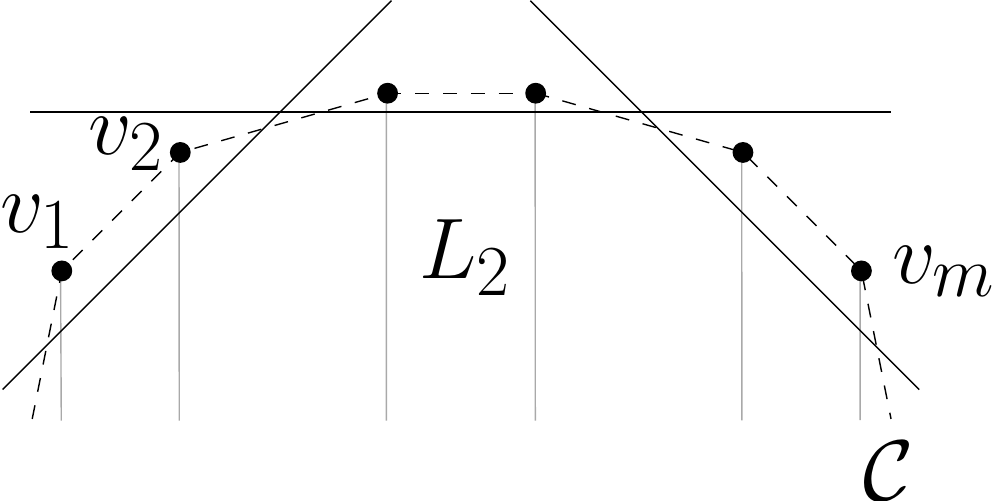}

   } 
     \centering\subfigure{ 
    \includegraphics[scale=.7]{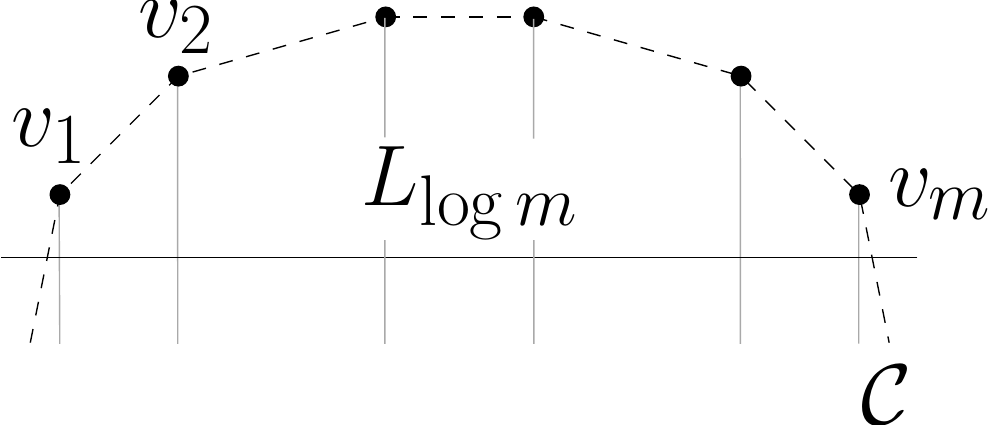}
   } 
     \caption{Sets of lines $L_j$.} 
       
     \label{fig:Ls} 
\end{figure}
\end{center}

Assume for now that $k$ is a multiple of $\beta + 1$,
and let $L$ consist of $k/(\beta + 1)$ copies of $L'$. We perturb the
lines in $L$ such that they are all distinct while
their relationship with the vertices of $\mathcal{C}$ remains unchanged.
It follows that $L$ has exactly $n \eqdef (2m-1)k/(\beta + 1)$ lines, with 
exactly $k$ lines in each conflict set $L_{v_i}$ (recall that by definition
$\beta = \log m$).

By Proposition~\ref{prop:dualpartition}, there is a coloring of 
$L$ such that each color class 
has size at most $k$ and such that each conflict set contains at most
$a(n,k)+1$ colors. The structure of $L$ lets us interpret this
coloring as follows:
let $T$ be a complete binary tree with 
$2m-1$ nodes and height $\beta$. We label the leaves of $T$ with 
the vertices $v_1$, $\ldots$, $v_m$, from left to right. 
Thus, every node $w$ of $T$ corresponds to an interval of
consecutive vertices of $\mathcal{C}$, namely the leaves
of the subtree rooted in $w$.
By assigning to $w$ the lines that lie exactly below the
vertices in this interval,
we  obtain a partition of $L$ into sets of
size $k/(\beta + 1)$. This leads to an 
interpretation of shallow partitions as multi-colorings
of trees, see Fig.~\ref{fig:tree}.
\begin{figure}[htb] 
  \begin{center}
   \includegraphics[scale=1]{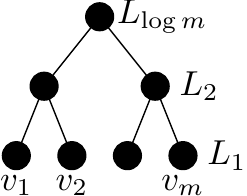}
  \end{center}
  \caption{A tree with height $\beta = \log n$. The leaves correspond to
  the vertices $v_i$, and the level of height $i$ corresponds to the lines
  in $L_i$. A line $\ell$ is stored in the node whose subtree corresponds to the
  vertices that have $\ell$ below them.}
  \label{fig:tree}
\end{figure}

\begin{prop}\label{prop:tree}
Let $T$ be a complete binary tree with height $\beta = \log m$ and $2m-1$ nodes,
and let $k$ be a multiple of $\log m + 1$.
Then there exists a multi-coloring of the nodes of $T$ with the
following properties: 
(i) every node is associated with a multiset of $k/(\beta + 1)$ colors;
(ii) each color class has at most $k$ elements; (iii) along each
root-leaf path there are at most $a(n,k)+1$ distinct colors,
where $n = (2m-1)k/(\beta + 1)$.
\end{prop}

\begin{proof}
Properties (i) and (ii) follow immediately from 
Proposition~\ref{prop:dualpartition} and the construction. 
For Property~(iii), observe that the lines encountered
along a root-leaf path are exactly the lines
below the vertex of $\mathcal{C}$ corresponding to the leaf.
\end{proof}

We can now prove the desired lower bound.

\begin{lemma}\label{lem:lowerb}
Let $T$ be a complete binary tree with height $\beta = \log m$ and $2m-1$ nodes, 
and let $k$ be a multiple of $\log m + 1$.
Consider a multi-coloring of $T$ such that
(i) every node is associated with a multiset of $k/(\beta+1)$
colors; and (ii) each color class has at most $2k$ elements.
Then there exists a root leaf-path with $\Omega(\log m/\log\log m)$ 
distinct colors.
\end{lemma}

\begin{proof}
We subdivide the nodes of $T$ into \emph{slices}. 
The first slice consists of the first $\lceil\log(3\beta)\rceil$ 
levels of $T$, 
the second slice consists of the following  $\lceil \log (6 \beta)\rceil$
levels, the third slice has the next $\lceil \log(9 \beta)\rceil$ levels,
and so on. In general, the $i$th slice consists of 
$\lceil \log (3i \beta)\rceil$ consecutive levels of $T$. 

We claim that there
exists a root-leaf path that has at least one distinct color 
for each slice that it crosses, except for the last one. 
To see this, we first consider a complete subtree $T'$ of $T$ that is
has its root in the first level of a slice $i$ and its
leaves in the last level of the same slice. As a complete
binary tree with $\lceil \log (3i \beta)\rceil$ levels, $T'$ has at least  
$3i \beta - 1 \geq 2i\beta + 2i$ nodes. Therefore, our 
multi-coloring needs to assign at least $2(i\beta +i)k/(\beta + 1)$ colors
in $T'$. Since each color class has size at most
$2k$, this requires at least
$i$ \emph{distinct} colors.

We now construct the required root-leaf path slice by slice.
Throughout, we maintain the invariant that after
$i$ slices have been considered, the path contains at least
$i$ distinct colors. This is certainly true at the root.
Now suppose that we have constructed a partial path $Q_{i-1}$ that  
ends at a node $z$ in the last level of the $(i-1)$th slice. If $Q_{i-1}$
contains at least $i$ distinct colors, we arbitrarily extend it
to a path $Q_{i}$ that ends at the bottom of the $i$th slice.
Otherwise, we pick an arbitrary child $z'$ of $z$. 
As noted above, the complete subtree that is rooted at $z'$ and restricted
to the $i$th slice contains at least $i$ distinct colors. Thus,
we can extend $Q_{i-1}$ through $z'$ to a path $Q_i$ that goes to
the bottom of the $i$th slice and that meets at least $i$ distinct colors.
The claim follows. 

It remains to calculate a lower bound for the number of slices $b$.
By construction, we must have
\[
  \sum_{i=1}^b \lceil \log (3i \beta)\rceil \geq \beta + 1.
\]
Now,
\begin{align*}
  \sum_{i=1}^b \lceil \log (3i \beta)\rceil &\leq
  \sum_{i=1}^b \log (4i \beta)\\
  &\leq b (2 + \log b + \log \beta)\\
  & \leq 3b \log\beta,
\end{align*}
since clearly $b \leq \beta$.
Hence, 
\[
b \geq \frac{\beta +1}{3\log\beta} = 
\Omega\Bigl(\frac{\log m}{\log\log m}\Bigr),
\]
as desired.
\end{proof}

We now indicate how to drop the assumption that $k$ is a multiple
of $\beta + 1$. Indeed, suppose that this is not the case,
but $k \geq \beta + 1$.  We first perform the
above construction with $k' \eqdef \lfloor k/(\beta + 1) \rfloor (\beta + 1)$ 
instead of $k$. Note that since $k \geq \beta+1$, we have $k \geq k'$. 
Then we add $k-k'$ suitably perturbed copies of $L_\beta$ 
(the set containing a line in conflict with all vertices of $\mathcal{C}$).
Let $L$ be the resulting set of lines. By Proposition~\ref{prop:dualpartition},
there exists a coloring of $L$ such that each color class has at most
$k \leq 2k'$ elements and such that each conflict set has at most
$a(|L|, k)+1$ distinct colors. The tree $T$ corresponding to $L$
has the same structure as before, but now each non-leaf node except
the leaf is associated with $k'/(\beta + 1)$ colors, while the leaves
have $k - k'$ additional colors. 
This suffices for the argument of Lemma~\ref{lem:lowerb} to go through.

\begin{theorem}
There is a constant $c > 0$ such that the following holds. For every $n$
and $k \in \{\log n, \ldots, n/4\}$, there 
exists a planar $n$-point set $P$
such that the crossing number for any $k$-partition of $P$ is at least
$c\log (n/k) / \log\log (n/k)$. Thus,
\[a(n,k) = \Omega\Biggl(\frac{\log (n/k)}{\log\log (n/k)}\Biggr).\]
\end{theorem}

\begin{proof}
Let $\beta \in \N$ be maximum with $(2^{\beta+1}-1)/(\beta+1) \leq n/2k$. 
Set $m \eqdef 2^\beta$ and $k' \eqdef \lfloor k/(\beta + 1) \rfloor (\beta + 1)$. 

From Propositions~\ref{prop:dualpartition} and \ref{prop:tree} 
and Lemma~\ref{lem:lowerb}, it follows that by taking the dual
we obtain a set $P'$ of 
$n' \eqdef (2m-1)k'/(\beta +1) + k - k'$ points such that
any $k$-partition of $P'$ has crossing number at least
$c' \log m / \log\log m$, for some constant $c' > 0$.

First note that $\beta < \log n$ and $k \geq \beta + 1$.
Hence, $k' \leq k \leq 2k'$ and $k - k' \leq \log n$. Thus, we can conclude that
\[
n' =  \frac{2m-1}{\log m + 1}k' + k - k' \leq \frac{n}{2} + \log n \leq n.
\]
and
\[
n' =  \frac{2m-1}{\log m + 1}k' + k - k' \geq \frac{n}{4k} \cdot \frac{k}{2} 
  = \frac{n}{8}.
\]
Thus, by adding at most $7n/8$ points that are contained in no $k$-shallow
halfplane, we can obtain from $P'$ a point set $P$ with $n$ points
and crossing number at least $c \log m / \log\log m$.
Finally, observe that 
\[
m \geq \frac{n'}{k'} - k \geq \frac{n}{9k}, 
\]
so $P$ also has crossing number at least
$c \cdot \frac{\log(n/k)}{\log\log(n/k)}$, for some $c > 0$. 
The result follows.
\end{proof}

Note that our construction also implies a similar lower bound in $\R^3$
by embedding the plane into three-dimensional space and perturbing
the points slightly. This provides an alternative proof of the result
by Afshani and Chan~\cite{AfshaniCh09}.

\section{Conclusion and Open Problems}

We have given a simple construction that give a lower bound
of $\Omega\Bigl(\frac{\log (n/k)}{\log\log (n/k)}\Bigr)$ for
the crossing number of
any shallow partition of a planar point set. \matousek's 
result gives an upper bound of $O(\log (n/k))$. Thus, there still
remains a factor of $\log\log(n/k)$ to be settled. Can we show that
\matousek's analysis is tight? Or, perhaps more interestingly, can
we construct shallow partitions with crossing number
$O\Bigl(\frac{\log (n/k)}{\log\log (n/k)}\Bigr)$?

\section*{Acknowledgments}

{\small
We would like to thank Nabil Mustafa for 
drawing our attention to shallow partitions
and for enlightening conversations.
We would also like to thank Timothy Chan
for answering our questions about shallow
partitions.}

\bibliographystyle{abbrv}
\bibliography{shallow}

\end{document}